\documentclass[pdftex,pagesize=auto,paper=letter,9pt,abstract=on,twocolumn]{scrartcl}

\usepackage{amsmath,mathtools} \usepackage{amsfonts} \usepackage{amsthm} \usepackage{graphicx} \usepackage{algorithm} \usepackage[algo2e, ruled, vlined]{algorithm2e} \usepackage[pdftex]{xcolor} \usepackage{url} \usepackage{textcomp,nicefrac} \usepackage{booktabs} \usepackage{wrapfig} \usepackage[format=plain,labelformat={simple},labelsep={period},labelfont={bf}]{caption} \usepackage{microtype} \usepackage[T1]{fontenc} \usepackage{lmodern}

\newtheorem{thm}{Theorem}[section] \newtheorem{theorem}{Theorem}[section]  \newtheorem{lemma}[thm]{Lemma}         

\newcommand{\ignore}[1]{} \newcommand{\notinproc}[1]{} \newcommand{\onlyinproc}[1]{#1} \newcommand{\skim}{SKIM\xspace} \newcommand{\tim}{TIM\xspace} \newcommand{\irie}{IRIE\xspace} \newcommand{\timplus}{TIM\textsuperscript{+}\xspace} \newcommand{\greedy}{GRE\xspace} \newcommand{\degree}{DEG\xspace}

\newcommand\E{\textsf{E}}   \newcommand{\kth}{\text{k}^{\text{th}}}   \newcommand{\INF}{\mathop{\sf Inf}}  \newcommand{\cascade}[1]{G^{(#1)}}  \newcommand{\instance}[1]{\textsf{#1}}

\SetAlCapHSkip{0pt} \SetAlgoCaptionSeparator{:} \DeclareMathOperator*{\argmax}{argmax}

\definecolor{darkgreen}{RGB}{0,100,0} \definecolor{orange}{RGB}{255,80,0}     \newcommand{\Xcomment}[1]{}

\title{Sketch-based Influence Maximization and Computation:\\Scaling up with Guarantees} \date{August 2014} \author{ EDITH COHEN\\Microsoft Research\\editco@microsoft.com \and DANIEL DELLING\\Microsoft Research\\dadellin@microsoft.com \and THOMAS PAJOR\\Microsoft Research\\tpajor@microsoft.com \and RENATO F.~WERNECK\\Microsoft Research\\renatow@microsoft.com }

\begin{document}

\maketitle

\begin{abstract} Propagation of contagion through networks is a fundamental process. It is used to model the spread of information, influence, or a viral infection. Diffusion patterns can be specified by a probabilistic model, such as Independent Cascade (IC), or captured by a set of representative traces.

Basic computational problems in the study of diffusion are \emph{influence queries}~(determining the potency of a specified \emph{seed set} of nodes) and \emph{Influence Maximization}~(identifying the most influential seed set of a given size). Answering each influence query involves many edge traversals, and does not scale when there are many queries on very large graphs. The gold standard for Influence Maximization is the greedy algorithm, which iteratively adds to the seed set a node maximizing the marginal gain in influence. Greedy has a guaranteed approximation ratio of at least~$(1-1/e)$ and actually produces a sequence of nodes, with each prefix having approximation guarantee with respect to the same-size optimum. Since Greedy does not scale well beyond a few million edges, for larger inputs one must currently use either heuristics or alternative algorithms designed for a pre-specified small seed set size.

We develop a novel sketch-based design for influence computation. Our greedy Sketch-based Influence Maximization~(\skim) algorithm scales to graphs with billions of edges, with one to two orders of magnitude speedup over the best greedy methods. It still has a guaranteed approximation ratio, and in practice its quality nearly matches that of exact greedy. We also present \emph{influence oracles}, which use linear-time preprocessing to generate a small sketch for each node, allowing the influence of any seed set to be quickly answered from the sketches of its nodes.\end{abstract}

\section{Introduction}

The spread of contagion~(information diffusion or spread of an infection) is a universal phenomenon that is extensively studied in the context of physical, biological, and social networks. Such cascades can have one or multiple sources (or {\em seeds}) and spread from infected nodes to neighbors through the link structure. A motivating application for the study of influence is viral marketing strategies~\cite{GLM:marketing2001,RichardsonDomingos:KDD2002}, in which the influence of a set $S$ of people in a social network is the number of adoptions triggered if we give $S$ free copies of a product. The problem also has important applications beyond social graphs, such as placing sensors in water distribution networks for detecting contamination~\cite{Leskovec:KDD2007}.

A popular model for information diffusion is {\em Independent Cascade} (IC), in which an independent random variable is associated with each (directed) edge $(u,v)$ to model the degree of influence of $u$ on $v$. A single propagation instance is obtained by instantiating all edge variables. We then study the distribution of a property of interest, such as the number of infected nodes, over these random instances.

The simplest and most studied IC model is {\em binary IC}, in which the range of the edge random variables is binary. A biased coin of probability $p_{uv}$ is flipped for each directed edge $(u,v)$. Accordingly, the edge can be either {\em live}, meaning that once $u$ is infected, $v$ is also infected, or {\em null}. This model was formalized in a seminal work by Kempe et al.~\cite{KKT:KDD2003} and is based on earlier studies by Goldenberg et al.~\cite{GLM:marketing2001}. Note that each direction of an undirected edge $\{u,v\}$ may have its own independent random variable, since influence is not necessarily symmetric. A particular propagation instance is specified by the set of live edges, and a node is infected by a seed set $S$ in this instance if and only if it is reachable from a seed node. The {\em influence} of $S$ is formally defined as the expectation, over instances, of the number of infected nodes.

Instead of working directly on this probabilistic IC model, Kempe et al.~\cite{KKT:KDD2003} proposed a simulation-based approach, in which a set $\{\cascade{i}\}$ of propagation instances (graphs) is generated in Monte Carlo fashion according to the influence model. The average influence of $S$ on $\{\cascade{i}\}$ is an unbiased estimate that converges to the expectation on the probabilistic model. The ability to compute influence with respect to an {\em arbitrary} set of propagation instances has significant advantages, as it is useful for instances generated from traces or by more complex models \cite{GRLK:KDD2010,ACKP:KDD2013}, which exhibit correlations between edges that cannot be captured by the simplified IC model~\cite{Gomez-RodriguezBS:ICML2011}. Moreover, the average behavior of a probabilistic model on a small set of instances captures its ``typical'' behavior, which is often more relevant than the expected value when the variance is very high.

 A basic primitive in the study of influence are {\em influence queries}: Compute (or approximate) the influence of a query set $S$ of seed nodes. With binary influence, this amounts to performing graph searches from the seed set in multiple instances. Unfortunately, this does not scale well when many queries are posed over graphs with millions of nodes.

Even more computationally challenging is the fundamental {\em Influence Maximization} problem, which is finding the most potent seed set of a certain size or cost. The problem was formalized by Kempe et al.~\cite{KKT:KDD2003} and inspired by Richardson and Domingos~\cite{RichardsonDomingos:KDD2002}. Kempe et al.~showed that, even when the influence function is deterministic (but the number $s$ of seeds is a parameter), the problem encodes the classic Max Cover problem and therefore is NP-hard~\cite{KKT:KDD2003}. Moreover, an inapproximability result of Feige~\cite{feige98} implies that any algorithm that can guarantee a solution that is at least~$(1-1/e+\epsilon)$ times the optimum is likely to scale poorly with the number of seeds. Chen et al.~\cite{CWW:KDD2010} showed that computing the exact influence of a single seed in the binary IC model, even when edge probabilities are $p=0.5$, is~\#P~hard~\cite{CWW:KDD2010}.

Using simulations, the objective studied by Kempe et al.~\cite{KKT:KDD2003} is then to find a set $S$ of seeds with maximum average influence over a fixed set of propagation instances. A natural heuristic is to use the set of most influential individuals, say those with high degree or centrality~\cite{KKT:KDD2003}, as seeds. This approach, however, cannot account for the dependence between seeds, missing the fact that two important nodes may ``cover'' essentially the same communities. Kempe et al.~\cite{KKT:KDD2003} proposed a greedy algorithm ({\sc Greedy}) instead. It starts with an empty seed set and iteratively adds to $S$ the node with maximum \emph{marginal} gain in influence (relative to current seed set). Since our objective is monotone and submodular, a classical result from Nemhauser et al.~\cite{submodularGreedy:1978} implies that the influence of the greedy solution with $s$ seeds is at least $1-(1-1/s)^s \geq 63\%$ of the best possible for any seed set of the same size. From Feige's inapproximability result, this is the best approximation ratio guarantee we can~(asymptotically and realistically) hope for.

{\sc Greedy} has become the gold standard for influence maximization, in terms of the quality of the results. {\sc Greedy}, however, does not scale to modern real-world social networks. The issue is that evaluating the marginal contribution of each node requires a directed reachability computation in each instance~(of which there can be hundreds). Several performance improvements to {\sc Greedy} have thus been proposed. Leskovec et al.~\cite{Leskovec:KDD2007} proposed CELF, which are ``lazy'' evaluations of the marginal contribution, performed only when a node is a candidate for the highest marginal contribution. Chen et al.~\cite{CWY:KDD2009} took a different approach, using the reachability sketches of Cohen~\cite{ECohen6f} to speed up the reevaluation of the marginal contribution of all nodes. While effective, even with these and other accelerations~\cite{CELFpp:WWW2011,OAYK:AAAI2014}, the best current implementations of {\sc Greedy} do not scale to networks beyond~$10^6$ edges~\cite{CWW:KDD2010}, which are quite small by modern standards.

To support massive graphs, several studies proposed algorithms specific to the IC model, which work directly with the edge probabilities instead of with simulations and thus can not be reliably applied to a set of arbitrary instances. Borg et al.~\cite{BBCL:SODA2014} recently proposed an algorithm based on reverse reachability searches from sampled nodes, similar in spirit to the approach used for reachability sketching~\cite{ECohen6f}. Their algorithm provides theoretical guarantees on the approximation quality and has good asymptotic performance, but large ``constants.'' Very recently, Tang et. al. \cite{TXS:sigmod2014} developed TIM, which engineers the (mostly theoretical) algorithm of Borgs et al.~\cite{BBCL:SODA2014} to obtain a scalable implementation with guarantees. A significant drawback of this approach is that it only works for a pre-specified seed set size $s$, whereas {\sc Greedy} produces a sequence of nodes, with each prefix having an approximation guarantee with respect to the same-size optimum. In applications we are often interested not in a single point, but in a trade-off curve that allows us to find a sweet spot of influence per cost or characterize the network. TIM also scales very poorly with the seed set size $s$, and the evaluation only considered seed sets of up to 50 nodes.

The DegreeDiscount~\cite{CWY:KDD2009} heuristic refines the natural approach of adding the next highest degree node. MIA~\cite{CWW:KDD2010} converts the binary IC sampling probabilities $p_e$ to deterministic edge weights and works essentially with one deterministic instance. IRIE, by Jung et al.~\cite{JHC:ICDM2012}, is a heuristic approximation of greedy addition of seed nodes, and has the best performance we are aware of for an algorithm that produces a sequence of seed nodes. In each step, the probability of each node to be covered by the current seed set $S$ is estimated using another algorithm (or simulations). They then use eigenvector computations to approximate marginal contributions of all nodes. Of those approaches, the IRIE heuristic scales much better and is much more accurate than other heuristics. In particular, it performs nearly as well as {\sc Greedy} on many research collaboration graphs~\cite{JHC:ICDM2012}.

\paragraph{Contributions.}

We design a novel sketch-based approach for influence computation which offers scalability with performance guarantees. Our main contribution is \skim\ (SKetch-based Influence Maximization), a highly scalable (approximate) implementation of the greedy algorithm for influence maximization. We also introduce {\em influence oracles}: after preprocessing that is almost linear, we can answer {\em influence queries} very efficiently, considering only the sketches of the query seed set.

We can apply our design on inputs specified as a fixed set of propagation instances, as in Kempe et al.~\cite{KKT:KDD2003}, with influence defined as the average over them. We also handle inputs specified as an IC model, where influence is defined as the expectation. Our model is defined precisely in Section~\ref{model:sec}. 

We now provide more details on our design. The exact computation of an influence query requires expensive graph searches from the query seed set $S$ on each of $\ell$ instances. The exact greedy algorithm for Influence Maximization requires a similar computation for each marginal contribution. We address this scalability issue by working with sketches.

The core of our approach are per-node summary structures which we call \emph{combined reachability sketches}. The sketch of a node compactly represents its influence ``coverage'' across~$\ell$ instances; we call this its \emph{combined reachability set}. The combined reachability sketch of a node, precisely defined in Section~\ref{sketch:sec}, is the bottom-$k$ min-hash sketch~\cite{bottomk07:ds,ECohenADS:PODS2014} of the combined reachability set of the node. This generalizes the reachability sketches of Cohen~\cite{ECohen6f}, which are defined for a single instance. The parameter $k$ is a small constant that determines the tradeoff between computation and accuracy. Bottom-$k$ sketches of sets support cardinality estimation, which means that we can estimate the influence~(over all instances) of a node or of a set of nodes from their combined reachability sketches. The estimate has a small relative error and good concentration~\cite{ECohen6f}. Our use of combination sketches and state-of-the-art optimal estimators is key to obtaining the best balance between sketch size and accuracy.

Our \skim\ algorithm for influence maximization is presented in Section \ref{binaryIM:sec}. It scales by running the greedy algorithm in ``sketch space,'' always taking a node with the maximum \emph{estimated} (rather than exact) marginal contribution.

\skim\ computes combined reachability sketches, but only until the node with the maximum estimated influence is computed. This node is then added to the seed set. We then update the sketches to be with respect to a {\em residual} problem in which the node that is selected into the seed set and its ``influence'' are no longer present. \skim\ then resumes the sketch computation, starting with the residual sketches, but (again) stopping when a node with maximum estimated influence (in the current, residual, instance) is found. A new residual problem is then computed. This process is iterated until the seed set reaches the desired size. Since the residual problem becomes smaller with iterations, we can compute a very large seed set very efficiently. We also prove that the total overhead of the updates required to maintain the residual sketches is small. In particular, for a set $\{G^{(i)}\}$ of $\ell$ {\em arbitrary} instances, the algorithm can be run to exhaustion, producing a full permutation of the nodes in $O(\sum_{i\in [\ell]} |G^{(i)}|+m \epsilon^{-2}\log^2 n)$ time, where $m$ is the sum over nodes of the maximum indegree (over instances). For all $s\geq 1$, the first $s$ nodes we select have with a very high probability (at least $1-1/n^c$ for a constant $c$) influence that is at least $1-(1-1/s)^s-\epsilon$ times the maximum influence of a seed set of the same size $s$. These are worst-case bounds. We propose an adaptive approach that exploits properties of actual networks, in particular a skewed influence distribution, to achieve faster running times with the same guarantees.

Our use of the residual instances by \skim\ is the key for maintaining the accuracy of the greedy selection through the execution and providing with high probability, approximation ratio guarantees that nearly match those of exact {\sc Greedy}. 

Section~\ref{binaryQ:sec} presents our influence oracles, which preprocess the input to compute combined reachability sketches for all nodes. For instances $\{\cascade{i}\}$ with $n$ nodes and $m^{(i)}$ edges, the sketches are built in $O(k \sum_i m^{(i)})$ total time. The influence of a set $S \subseteq V$ can then be approximated from the sketches of the nodes in $S$. The oracle applies the union cardinality estimator of Cohen and Kaplan~\cite{CK:sigmetrics09} to estimate the union of the influence sets of the seed nodes. The query runs in time $O(|S|k\log |S|)$ and unbiasedly with a well-concentrated relative error of $\epsilon=1/\sqrt{k}$. While preprocessing depends on the number of instances, the sketch size and the approximation quality only depend on the sketch parameter $k$.

The asymptotic bounds we obtain are novel also from a theoretical perspective, and significantly improve the state of the art, even for influence maximization on a single~(deterministic) instance (select a seed set in a directed graph with maximum reachable set).

Section~\ref{experiments:sec} presents an extensive experimental study. Besides demonstrating the scalability of our algorithms on real-world networks, we compare \skim\ with existing approaches, including exact {\sc Greedy} (when size allows), the state-of-the-art IRIE heuristic, and TIM. We obtain IC models from networks by using the well-studied weighted and uniform~\cite{KKT:KDD2003} probabilities. Our algorithms scale up to very large graphs with barely any compromise on quality over exact {\sc Greedy}, with theoretical guarantees. On instances generated by an IC model, we achieve more than an order of magnitude speedup over the best greedy heuristics, which are designed specifically for this model. Even for a fixed small seed set size, \skim\ is significantly faster than TIM.

Moreover, our algorithm is efficient and accurate enough to be executed exhaustively, producing a full permutation of the nodes for networks with billions of edges. For the first time, we provide the full (approximate) Pareto front of influence versus seed set size. These relations showcase a basic property of the network, and the general pattern that a small fraction of nodes influences a large fraction of the network. In contrast, most previous studies we are aware of only considered seed sets with at most 50 nodes, revealing only a very restricted view of this relation.

\section{Model} \label{model:sec}

A {\em propagation instance} $G=(V,E)$ is specified by the edge set $E$. The {\em influence} of a set of nodes $S$ in instance $G$ is the number of nodes reachable from $S$ using the edges $E$: \begin{equation}\INF(G,S) = |\{ u \mid S \leadsto u \}|,\end{equation} where the predicate $S \leadsto u$ holds if $u\in S$ or if there is a forward path from a node in $S$ to the node $u$.

Our input is specified as a set $\mathcal{G}=\{\cascade{i}\}$ of $\ell\geq 1$ propagation instances $\cascade{i}=(V,E^{(i)})$ on the same set of nodes. The influence of $S$ over all instances $\{\cascade{i}\}$ is the average single-instance influence: \begin{equation} \label{timedinf} \INF(\mathcal{G},S)=\INF(\{\cascade{i}\},S) = \frac{1}{\ell} \sum_{i\in [\ell]} \INF(\cascade{i},S). \end{equation} The set of propagation instances can be derived from cascade traces or generated by a probabilistic model.

The input can also be specified as a probabilistic model, such as Independent Cascade (IC)~\cite{KKT:KDD2003}, which defines a distribution $\mathcal{G}$ over instances $G \sim \mathcal{G}$ that share a set $V$ of nodes. In this case, the influence of $\mathcal{G}$ is defined as the expectation \begin{equation} \label{timedprobinf} \INF(\mathcal{G},S)=\E_{G\sim \mathcal{G}} \INF(G,S). \end{equation}

We are interested in {\em influence oracles} and in {\em influence maximization}. Influence queries are specified by a seed set~$S \subset V$ and the goal is to compute (or estimate) the influence $\INF(\mathcal{G},S)$. Influence oracles, after efficient preprocessing of the input, allow us to support very fast queries. Influence maximization is the problem of finding a seed set~$S\subset V$ with maximum influence, where $|S| = s$ is given. We are interested in efficiently computing a seed set whose influence is close to the maximum one, as well as in computing a sequence of seeds so that each prefix has influence that is close to maximum for its size.

\section{Combined Reachability Sketches} \label{sketch:sec} At the heart of our approach are {\em combined reachability sketches}, which are summary structures $X_u$ that we associate with each node $u$. The combined sketches can be defined with respect either to a set $\mathcal{G}=\{\cascade{i}\}$ of $\ell\geq 1$ instances or to a probabilistic model~$\mathcal{G}$.

We first consider as input a set of $\ell\geq 1$ instances. We define the {\em reachability set} of a node $u$ in instance $G$ as~$R(G,u) = \{ v \mid u \leadsto_G v \},$ where $u\leadsto_G v$ means that $v$ is reachable from~$u$ in $G$. Considering all instances, the {\em combined reachability set} is a set of node-instance pairs: $R_u = \{(v,i) \mid u\leadsto_{G^{(i)}} v \}.$ The influence of a set of nodes $S$ on instances~$\{\cascade{i}\}$ can thus be expressed as \begin{equation} \label{InfbinaryIC:eq} \INF(\{\cascade{i}\},S)= \frac{1}{\ell} \sum_{i \in [\ell]} \Bigl| \bigcup_{u\in S} R(G^{(i)},u) \Bigr| = \frac{1}{\ell} \Bigl| \bigcup_{u\in S} R_u \Bigr|. \end{equation} This is the average over the instances $\{\cascade{i}\}$ (with $i\in [\ell]$) of the number of nodes reachable from at least one node in $S$.

The combined reachability sketch of a node captures its reachability information {\em across} instances. The sketches we use are the bottom-$k$ min-hash sketches~\cite{ECohen6f,bottomk07:ds} $X_v$ of the combined reachability sets $R_v$: We associate with each node-instance pair $(v,i)$ an independent random rank value~$r^{(i)}_v \sim U[0,1]$, where~$U[0,1]$ is the uniform distribution on $[0,1]$. The \emph{combined reachability sketch} of $u$ is the set of the $k$ smallest rank values amongst~$\{r^{(i)}_v \mid (v,i)\in R_u\}$: \begin{equation} \label{combinedrs:eq} X_u = \text{\textsc{Bottom-}$k$}\{ r^{(i)}_v \mid (v,i) \in R^{(i)}_u\}, \end{equation} where \textsc{Bottom-}$k$ of a set is its subset consisting of the~$k$ smallest values. When there is a single instance ($\ell=1$) the combined reachability sketches are the same as the reachability sketches of Cohen~\cite{ECohen6f}.

We define the \emph{threshold rank} $\tau_u$ of each node $u$ as \begin{equation} \label{combinedthresh:eq} \tau_u=\kth\bigl(\{r^{(i)}_v \mid (v,i) \in R^{(i)}_u \}\bigr), \end{equation} which is the $k$th lowest rank value in $R_u$. (For a set $Y$ of cardinality $|Y|< k$, we define $\kth(Y)\equiv 1$.) Therefore, when~$|X_u| = k$ we have~$\tau_u=\max\{X_u\}$, and $\tau_u=1$ otherwise. The cardinality~$|R_u|$ can be estimated from $X_u$ using a bottom-$k$ cardinality estimator. The estimate is~$|X_u|$ if~$\tau_u=1$~(i.e., if $|X_u|<k$) and is~$(k-1)/\tau_u$ otherwise. This estimate has a Coefficient of Variation (CV), which is the ratio of the standard deviation to the mean, that is never more than~$1/\sqrt{k-2}$ and is well concentrated~\cite{ECohen6f}. By applying Chernoff bounds with $c>1$, we obtain that using~$k = (2+c)\epsilon^{-2} \ln n$, the probability of having relative error larger than $\epsilon$ is at most~$1/n^c$. Therefore, we can be correct with high probability on estimating the influence of all nodes.

 \subsection{Structured Permutation Ranks} \label{permranks:sec}\label{structranks:sec} Instead of using ranks drawn from $U[0,1]$, we can work with integral permutation ranks with respect to a permutation on the $n\ell$ node-instance pairs. We can also structure the permutation so that each sequence in positions $in+1$ to $(i+1)n$ for integral $i\geq 0$ has each node appear in exactly one pair. The associated instance with a node $v$ in chunk $i$ is randomly selected from instances $j$ for which the pair $(v,j)$ does not have a permutation rank of $in$ or less (independently for each node). One can show that this can only improve estimation accuracy~\cite{ECohenADS:PODS2014}. Only the first~$\min\{k,\ell\} n$ positions can be included in combined reachability sketches of nodes.

When estimating influence, we can convert permutation ranks to random ranks using the exponential distribution~\cite{ECohen6f}\onlyinproc{.} We can also estimate cardinality of a subset of the $D=n\ell$ elements directly from permutation ranks $[D]$, using the unbiased estimator~$1+(k-1)(D-1)/(T-1)$, where the threshold $T$ is the $k$th smallest permutation rank. This estimator can be interpreted as setting aside the element with permutation rank $T$, and estimating the fraction (of the other $D-1$ elements) that is in our set by the fraction of such elements with rank smaller than $T$, which is~$(k-1)/(T-1)$.

\subsection{Sketches for an IC Model} \label{ICsketches:Sec}

We now define sketches with respect to a binary IC model $\mathcal{G}$, presented as a graph with probabilities $p_e$ associated with its edges. The influence of a set of nodes $S$ is \begin{equation} \label{ICmodelinf} \INF(\mathcal{G},S)=\E_{G\sim \mathcal{G}} \bigl|\bigcup_{u\in S} R(G,u)\bigr|. \end{equation} The sketches we define for $\mathcal{G}$ also contain at most $k$ rank values, but provide approximation guarantees with respect to \eqref{ICmodelinf}. The sketches can be interpreted as the sketches computed for $\ell$ instances generated according to the model $G \sim \mathcal{G}$ as $\ell \rightarrow \infty$. When doing so, at the limit, each unique rank value corresponds to a unique instance, so we do not need to explicitly represent ``instances.'' We work with structured permutation ranks (Section~\ref{permranks:sec}). Since it suffices to consider the first $k n$ ranks, this conveniently removes the dependence of the rank representation on $\ell$. We can similarly apply an estimator to the $k$th smallest rank $T\leq kn-k$ to estimate influence: Instead of estimating cardinality (which goes to infinity with $\ell$) and dividing by $\ell$ using the estimator $\frac{1}{\ell}+\frac{(k-1)(n\ell-1)}{\ell(T-1)}$ we take the limit as $\ell \rightarrow \infty$ and estimate influence using $n(k-1)/(T-1)$.

\section{SKIM: Sketch Space IM} \label{binaryIM:sec}

In this section we present our Sketch-based Influence Maximization~(\skim) algorithm. We first review {\sc Greedy}, the greedy algorithm for influence maximization (working with $\ell$ instances) presented by Kempe et al.~\cite{KKT:KDD2003}. {\sc Greedy} is applied with respect to the influence objective $\INF(\mathcal{G},S)$, as defined in Equation~\eqref{timedinf}. It starts with an empty seed set $S=\emptyset$. In each iteration, it adds to~$S$ the node $v$ with maximum {\em marginal gain}, \begin{equation} \label{binarymarginal} \INF(\mathcal{G},S \cup \{v\}) - \INF(\mathcal{G},S)=\frac{1}{\ell} \Bigl|\bigcup_{u\in S\cup \{v\}} R_u \setminus \bigcup_{u\in S} R_u \Bigr|. \end{equation} This is the same as choosing $v$ maximizing $\INF(\mathcal{G},S\cup\{v\})$.

\skim\ approximates exact {\sc Greedy} by ensuring that at {\em each iteration}, with sufficiently high probability, or in expectation over iterations, the node we choose to add to the seed set has a marginal gain that is close to the maximum one. To do so, it suffices to compute sketches only to the point that the node with the maximum estimated marginal gain is revealed. To maintain accuracy, we maintain a residual problem and respective sketches.

 \skim\ constructs (partial) combined reachability sketches by adapting a construction of reachability sketches~\cite{ECohen6f}: It processes node-instance pairs $(u,i)$ by increasing rank, performing a reverse reachability search in $G^{(i)}$ from~$u$. The sketch $X_v$ of each visited node $v$ is augmented with the rank $r^{(i)}_u$ of the pair. For a given value of $k$, the first node~$u$ whose sketch reaches size $k$ is also the node with maximum estimated influence. This is because the bottom-$k$ cardinality estimate of a node depends only on the $k$th smallest rank in~$X_u$, $\tau_u$ (which is a complete sufficient statistic for cardinality estimation from the sketch~\cite{ECohenADS:PODS2014}); see Equation~\eqref{combinedthresh:eq}. For the node $u$, $\tau_u$ is equal to the rank $r^{(i)}_u$ of the last processed pair $(u,i)$. For other nodes $v$ with incomplete sketches, we know that $\tau_v \geq r^{(i)}_u$, so their estimate is lower.

Sketch building is suspended once the node $v$ with maximum estimated influence is found. \skim then adds~$v$ to the seed set and generates a {\em residual} problem, with $v$ and all node-instance pairs it covers removed from the instances~$\mathcal{G}$. The~(partially computed) sketches of each remaining node $u$ are updated using~$X_u \gets X_u \setminus X_v$, which deletes from the sketch the ranks of all covered node-instance pairs.

The process of building sketches is then resumed on the residual problem, working with updated partial sketches and instances. We continue processing node-instance pairs in increasing rank order, starting from the first rank that exceeds $\tau_v$ and skipping pairs that are already covered.

\begin{algorithm2e}[t] \caption{Sketch-based Influence Maximization\label{greedySSC:alg}} \DontPrintSemicolon \SetKw{True}{true} \SetKw{False}{false} \SetKw{Continue}{continue} \SetKw{Skip}{skip} \SetKw{Prune}{prune} \SetKwFunction{Return}{return} \SetKwFunction{Append}{append} \SetKwArray{Covered}{covered} \SetKwArray{Index}{index} \SetKwArray{Size}{size} \SetKwArray{SeedList}{seedlist} \SetKwData{Rank}{rank} \tcp{Initialization} \lForAll{pairs $(u,i)$}{$\Covered{u,i} \leftarrow \False$} \lForAll{nodes~$v$}{$\Size{v} \leftarrow 0$} $\Index \leftarrow$ hash map of node-instance pairs to nodes\; $\SeedList \leftarrow \emptyset$\tcp*{List of seeds \& marg.\ influences} $\Rank \leftarrow 0$\;

\BlankLine shuffle the $n\ell$ node-instance pairs~$(u,i)$\; \BlankLine

\tcp{Compute seed nodes} \While{$|\SeedList| < n$}{

\While(\tcp*[f]{Build sketches}){$\Rank < n\ell$}{ $\Rank \leftarrow \Rank +1$\; $(u,i)\leftarrow$ $\Rank$-th pair in shuffled sequence\; \BlankLine \If{$\Covered{v,i}=\False$}{ BFS from~$u$ in reverse graph~$G^{(i)}$, during which\; \ForEach{scanned node~$v$}{

$\Size{v} \leftarrow \Size{v} + 1$\; $\Index{u,i} \leftarrow \Index{u,i} \cup \{v\}$\; \If{$\Size{v} = k$}{ $x \leftarrow v$\tcp*{Next seed node} abort sketch building\; } } } }

\BlankLine \If{all nodes $u$ have $\Size{u}<k$ }{ $x \leftarrow \argmax_{u \in V}{\Size{u}}$\; }

\BlankLine $I_x \leftarrow 0$\tcp*{The coverage of~$x$} \ForAll(\tcp*[f]{Residual problem}){instances~$i$}{ (forward) BFS from~$x$ in graph~$G^{(i)}$, during which\; \ForEach{scanned node~$v$}{ \lIf{$\Covered{v,i}$}{prune} $I_x \leftarrow I_x + 1$\; $\Covered{v,i} \leftarrow \True$\tcp*{Cover $v$ in $i$} \ForAll{nodes $w$ in $\Index{v,i}$}{ $\Size{w} \leftarrow \Size{w} - 1$\; } $\Index(v,i) \leftarrow \bot$\tcp*{Erase $(v,i)$ from $\Index$} } } \BlankLine $I_x \leftarrow I_x /\ell$\; \SeedList.\Append{$x$,$I_x$}\; } \Return{\SeedList}\; \end{algorithm2e}

We provide pseudocode for \skim\ as Algorithm \ref{greedySSC:alg}. Instead of maintaining the actual partial sketches $X_v$, the algorithm only keeps their cardinalities $\mbox{\sf size}[v]$. To support correct and efficient updates of the sketches, we maintain an inverted index $\mbox{\sf index}[u,i]$ that lists, for each rank value $r^{(i)}_u$ we processed, all nodes $v$ such that $r^{(i)}_u \in X_v$. The entry for rank~$r^{(i)}_u$ is created and populated when we perform a reverse reachability search from pair $(u,i)$. The algorithm outputs the list \SeedList of pairs $(\sigma_i,I_i)$, where $\{\sigma_i\}$ is a permutation of the nodes according to the order they are selected into the seed set, and $I_i$ is the marginal influence of $\sigma_i$. The surprising property of our construction is that this whole iterative process is very efficient. If we run \skim\ with a fixed~$k= c \epsilon^{-2} \log n$, Section~\ref{sec:analysis} will show that we obtain the following worst-case performance guarantees: \begin{theorem} \label{skimtime:thm} \skim\ runs in time $O(n\ell+\sum_i |E^{(i)}| + m \epsilon^{-2} \log^2 n)$, where $m= \sum_v \max_i \text{{\sc InDeg}}^{(i)}(v)\leq |\bigcup_i E^{(i)}|$. The permutation $\{\sigma_i\}$ of nodes has the property that with probability $1-1/n^{\Omega(c)}$, for all $s\in [n]$, the set of seed nodes $S=\{\sigma_1,\ldots,\sigma_s\}$, has $\INF(\{\cascade{i}\},S) \geq (1-1/e-\epsilon) \arg\max_{Z \mid |Z|\leq s} \INF(\{\cascade{i}\},Z)$. \end{theorem} 

\subsection{Algorithm Analysis} \label{sec:analysis}

\subsubsection{Correctness}

It is not hard to show that the influence of a node $v$ in the residual problem of iteration $i$ is equal to its marginal influence with respect to $S=\{\sigma_1,\ldots,\sigma_{i-1}\}$ in the original problem. Therefore,~$I_i$, which is the influence of $\sigma_i$ in the residual problem of iteration $i$, is the marginal influence of $\sigma_i$, with respect to~$S=\{\sigma_1,\ldots,\sigma_{i-1}\}$ in the original problem. Thus, by definition, for all $s\in [n]$ and~$S=\{\sigma_1,\ldots,\sigma_s\}$, $\INF(\{\cascade{i}\},S)=\sum_{i\in [s]} I_i$.

We also show that the partial sketches correctly capture a component of the sketches computed for the residual problem: \begin{lemma} At the end of an iteration selecting $v$, each updated partial sketch $X_u$ is equal to the set of entries of the combined reachability sketch $X'_u$ of $u$ in the residual problem that have rank value at most $\tau_v$. \end{lemma} \begin{proof}[Proof sketch] The content of each sketch $X_u$ before computing the residual is clearly a superset of all reachable node-instance pairs~$(z,i)$ with rank $r^{(i)}_z\leq \tau_v$ in the residual problem. We can then verify that entries are removed from~$X_u$ only and for all covered node-instance pairs with~$r^{(i)}_z\leq \tau_v$. \end{proof}

 \subsubsection{Running Time}

We now analyze the running time of \skim. All updates of the residual problem together take time linear in the size of~$\{\cascade{i}\}$, since nodes and edges that are covered by the current seed set are removed once visited and never considered again. The remaining component of the computation is determined by the number of times ranks are inserted (and removed) from sketches. Inserting a value to $X_u$ involves a scan of all (remaining) incoming edges to $u$ in an instance. Removals of ranks can be charged to insertions. So we need to bound the total number of rank insertions: \begin{lemma} The expected total number of rank insertions at a particular node is $O(k \ln n)$. \end{lemma} \begin{proof}[Proof sketch] Consider a sketch $X_v$. We can show, viewing the sketches as uniform samples of reaching pairs, that each rank value removal corresponds to cardinality---and hence influence~(marginal gain)---being reduced in expectation by a factor of $1-1/k$. The initial influence is at most $n$, so there are at most~$k\ln (n k) $ insertions until the marginal influence is reduced below $1/k$, at which point we do not need to consider the node. \end{proof} The running time is dominated by the sum over nodes $v$, of the number of times a rank is inserted to the sketch of~$v$, times the in-degree of $v$ (the maximum over instances). From the lemma, we obtain a bound of $O(k m \ln n)$ on the total number of insertions. Thus, we obtain a bound of $O(km\ln ( n) + \sum_i |\cascade{i}|)$ on the running time of the algorithm.

\subsubsection{Approximation Ratio} To obtain an approximation that is within $1+\epsilon$ with good probability, we can choose a fixed $k= c \epsilon^{-2}\log n$, for some constant~$c$. The relative error of each influence estimate of a node in an iteration is at most $\epsilon$ with probability of at least~$1-1/n^c$. Since we use polynomially many estimates (maximize influence among~$n$ nodes in each of at most $n$ iterations), all estimates are within a relative error of $\epsilon$ with probability that is polynomially close to $1-1/n^{c-2}$. Lastly, we bound the approximation ratio of the ``approximate'' greedy algorithm we work with, which uses seeds with close to maximum instead of maximum marginal gain: \begin{lemma}\label{approxgreedy:lemma} With any submodular and monotone objective function, approximate greedy, which iteratively chooses a node with marginal gain that is at least $(1-\delta)$ of the maximum, has an approximation ratio of at least $(1-(1-1/s)^s-O(\delta))$. The same claim holds in expectation when the selection is well concentrated, that is, its probability of being below $(1-a\delta)$ times the maximum decreases exponentially with $a>1$. \end{lemma} \begin{proof} The argument extends the analysis of exact greedy by Nemhauser et al.~\cite{submodularGreedy:1978}. For any $s$, and after selecting any set $U$ of seeds, the maximum marginal gain by adding a single node is always at least~$1/s$ of the maximum possible gain for $s$ nodes. When using the approximation, this is at least $(1-\delta)/s$ of the maximum possible gain. Therefore, after approximate greedy selection of $s$ nodes, the influence is at least $1-(1-(1-\delta)/s)^s \leq 1-(1-1/s)^s -O(\delta)$ using the first order term of the Taylor expansion. \end{proof}

\subsection{Extensions}

\subsubsection{Adaptive Error Estimation} This worst-case analysis is too pessimistic, both for the approximation ratio and running time. In our experiments, we tested \skim\ with a fixed $k$, and observed that the computed seed sets had influence that is much closer to the exact greedy selection than indicated by the worst-case bounds.

The explanation is that the influence distribution on real inputs is heavy-tailed, with the vast majority of nodes having a much smaller influence than the one of maximum influence. One factor of $O(\log n)$ in the worst-case running time is due to a ``union bound'' ensuring a relative error of $\epsilon$ for all nodes in all iterations, with high probability. With a heavy tail distribution, we can identify the maximum with a small error if we ensure a small error only on the few nodes that have influence close to the maximum. Furthermore, when the maximum influence is separated out from other influence values, our approximate maximum is more likely to be the node with actual maximum influence. Moreover, the estimation error over iterations averages out, so as the seed set gets larger we can work with lower accuracy and still guarantee good approximation. 

We propose incorporating error estimation that is {\em adaptive} rather than worst-case. This facilitates tighter confidence bounds on the estimation quality of our output. It also allows us to adjust the sketch parameter $k$ during computation in order to meet pre-specified accuracy and confidence levels.

Let the {\em discrepancy} in an iteration be the gap between the actual maximum and the marginal influence of the selected seed. We will bound the sum of discrepancies across iterations by maintaining a confidence distribution on this sum.

The estimation uses two components. (i)~The exact marginal influence $I_s$ of the selected node in each iteration, as well as the sum~$I=\sum_{i\leq s} I_s$, which is the influence of our seed set. The value~$I_s$ is computed when generating the residual problem. (ii)~Noting in each iteration the size of the second largest sketch~(excluding the last processed rank). Intuitively, if the second largest sketch is much smaller than the first one, it is more likely that the first one is the actual maximum. We bound the discrepancy in a single iteration using Chernoff bounds. The probability that the sum of independent Bernoulli trials falls below its expectation $\mu$ by more than~$\nu \mu$ is \begin{equation}\label{chernoff} \Pr[ Z < (1-\nu)\mu] < \left( \frac{\exp(-\nu)}{(1-\nu)^{(1-\nu)}} \right)^\mu. \end{equation} We use this to bound the probability that the discrepancy exceeds~$\Delta\epsilon$, where $\Delta$ is the exact marginal gain of our selected seed node. We consider the second largest sketch size,~$k'\leq k-1$~(the last rank of $\tau$ is not considered part of the sketch even if included). We use $Z= k'$, $\mu = \tau \Delta (1+\epsilon)$, and $\nu =1- \frac{k'}{\tau \Delta (1+\epsilon)}$ in Equation~\eqref{chernoff} to obtain a confidence level.

Finally, to maintain an upper bound on the confidence-error distribution of the sum of discrepancies, we take a convolution, after each iteration, of the current distribution with the distribution of the current iteration. 

\subsubsection{Alternative Implementations} \onlyinproc{ \skim\ can be adapted for higher concurrency by running the sketch-building phases in batches of ranks. We can also adapt it to process inputs presented as an IC model instead of as a set of instances. This yields a more efficient implementation than when generating a set of instances using simulations and running \skim\ on them. In IC-model \skim, the residual problem is a collection of partial models and sketch building is performed on the probabilistic model. We omit details due to space limitations. } 

\section{Influence Oracles} \label{binaryQ:sec}

We now present an accurate and efficient oracle for binary influence, which is based on precomputing a combined reachability sketch (as defined in Section~\ref{sketch:sec}) for each node. We preprocess a set of $\ell$ instances $\mathcal{G}=\{\cascade{i}\}$ using $O(k \sum_{i=1}^\ell |E^{(i)}|)$ computation and working storage of $O(k)$ per node. The preprocessing generates combined reachability sketches $X_v$ of size $O(k)$ for each node $v\in V$.

\begin{theorem} \label{binaryinforacles:thm} Given a set $\{X_v\}$ of combined reachability sketches for $\mathcal{G}$ with parameter $k$, influence queries $\INF(\mathcal{G},S)$ for a set $S$ of nodes can be estimated in $O(|S|k\log|S|)$ time from the sketches $\{X_u \mid u\in S\}$. The estimate is non\-negative and unbiased, has CV at least~$1/\sqrt{k-2}$, and is well concentrated, meaning that the probability that the relative error exceeds $a/\sqrt{k}$ decreases exponentially with $a>1$. \end{theorem}

We next present the two components of our oracle: estimating the influence of $S$ from the sketches of the nodes in $S$ and efficiently computing all combined reachability sketches.

\subsection{Influence Estimation from Sketches} We show how to use the combined reachability sketches of a set of nodes $S$ to estimate the influence of $S$, as given in Equation~\eqref{InfbinaryIC:eq}. In graph terms, this means estimating the cardinality of the union $\bigcup_{u\in S} R_u$ from the sketches $X_u$, with $u\in S$. The influence~$\INF(\mathcal{G},S)$ is the union cardinality divided by the number of instances $\ell$ and, accordingly, is estimated using $\widehat{\bigl|\bigcup_{v\in S} R_v\bigr|}/\ell$. Our estimators use the threshold rank $\tau_u$ of each node $u$; see~Equation~\eqref{combinedthresh:eq}.

From the bottom-$k$ sketches of each set $R_u$ for $u\in S$ we can unbiasedly estimate the cardinality of the union $\bigcup_{u\in S} R_u$. One way to do this is to compute the bottom-$k$ sketch of the union~\cite{ECohen6f}, which has threshold value $\tau=\kth\{\bigcup_{u\in S} X_u\}$ and apply the cardinality estimator $(k-1)/\tau$. This would already conclude the proof of Theorem \ref{binaryinforacles:thm}.

In our implementation, we use a strictly better union cardinality estimator that uses all the (at most $k |S|$) values in the set of sketches instead of just the $k$th smallest: \begin{equation} \label{unionest} \widehat{\bigl|\bigcup_{v\in S} R_v\bigr|}=\sum_{z \in \bigcup_{v\in S} X_v\setminus \{\tau_v\}} \frac{1}{\max_{u\in S | z\in X_u\setminus \{\tau_u\}} \tau_u}. \end{equation} This estimator, proposed by Cohen and Kaplan~\cite{CK:sigmetrics09}, can be computed from the $|S|$ sketches in time $O(|S|k\log|S|)$, by first sorting the $|S|$ sketches by decreasing threshold, and then identifying for each distinct rank value the threshold of the first sketch that contains it. When the sets $R_u$ are all the same, the estimate is the same as applying an estimator to the bottom-$k$ sketch on the union, but Equation~\eqref{unionest} can have up to a factor of $\sqrt{|S|}$ lower CV when the sets $R_u$ are sufficiently disjoint. Moreover, this estimator is an optimal sum estimator in that it minimizes variance given the information available in the sketches.

 We can also derive a permutation version of Equation~\eqref{unionest}. The simplest way is to treat the permutation rank $T$ as a uniform rank $r=(T-1)/(\ell n -1)$ which is the probability that the rank of another node is smaller than $T$.

\subsection{Building Combined Reachability Sketches}

When there is a single instance $G=(V,E)$, the combined sketches are simply reachability sketches~\cite{ECohen6f,bottomk07:ds}. Reachability sketches $Y_v$ for all nodes can be computed very efficiently, using at most $m k$ edge traversals in total, where $m$ is the number of edges~\cite{ECohen6f}. 

\begin{algorithm2e}[t] \caption{Combined reachability sketches\label{reachsketch1:alg}} \DontPrintSemicolon \SetKwArray{Sketches}{sketches} \SetKwArray{LocalSketches}{local} \SetKwData{Rank}{rank} \SetKw{Skip}{skip} \SetKw{Prune}{prune} \SetKwFunction{Merge}{merge} \ForAll{nodes~$u \in V$}{ $\Sketches{u} \leftarrow \emptyset$\tcp*{Global sketches} $\LocalSketches{u} \leftarrow \emptyset$\tcp*{Instance-local sketches} } \BlankLine shuffle the $n\ell$ node-instance pairs~$(u,i)$\; \BlankLine

\ForAll{instances~$i$}{ \tcp{Build local sketches for instance~$i$} \For{pairs~$(u,j)$ with~$j=i$ by increasing rank~$r$}{ BFS from~$u$ in reverse graph~$G^{(i)}$, during which\; \ForEach{scanned node~$v$}{ \lIf{$|\LocalSketches{v}| = k$}{prune} $\LocalSketches{v} \leftarrow \LocalSketches{v} \cup \{r\}$\; } } \BlankLine \tcp{Merge local sketches into global sketches} \ForAll{nodes~$u$}{ \tcp{Both $\Sketches{u}$ and~$\LocalSketches{u}$ are sorted} $\Sketches{u} \leftarrow \Merge{$\Sketches{u}$,$\LocalSketches{u}$}$\; trim~$\Sketches{u}$ to size~$k$\; $\LocalSketches{u} \leftarrow \emptyset$\; } } \Return{\Sketches}\; \end{algorithm2e}

 Algorithm \ref{reachsketch1:alg} computes combined sketches by applying the pruned searches algorithm of Cohen~\cite{ECohen6f} on each instance $\cascade{i}$, obtaining a sketch $Y^{(i)}_v$ for each node, and combining the results. The combined sketch $X_v$ is obtained by taking the bottom-$k$ values in the union of the $\ell$ sketches, defined as~$X_v \gets \text{{\sc bottom-}}k(\cup_{i\in \ell} Y^{(i)}_v ).$ 

The algorithm runs in $O(k \sum_i |E^{(i)}|)$ time. Rather than storing all sets of sketches, we can compute and merge concurrently or sequentially, but after each step, take the bottom-$k$ values in the current bottom-$k$ set and the newly computed sketch for instance~$G^{(i)}$: $X_v \gets \text{{\sc bottom-}}k\{X_v,Y^{(i)}_v\}$. Therefore, the additional run time storage requirement for sketches is~$O(nk)$. This gives us the worst-case bounds on the computation stated in Theorem~\ref{binaryinforacles:thm}. \section{Experiments} \label{experiments:sec}

We implemented our algorithms in C++ using Visual Studio 2013 with full optimization. All experiments were run on a machine with two Intel Xeon E5-2690 CPUs and 384\,GiB of DDR3-1066 RAM, running Windows 2008R2 Server. Each CPU has 8 cores~(2.90\,GHz, 8\,$\times$\,64\,kiB L1, 8~$\times$\,256\,kiB, and 20\,MiB L3 cache), but all runs are sequential for consistency.

We ran our experiments on benchmark networks available as part of the SNAP~\cite{SNAP} and WebGraph~\cite{bv-twfct-04} projects. More specifically, we test \emph{social}~(\instance{Epinions}, \instance{Slashdot}, \instance{Gowalla}, \instance{TwitterFollowers}, \instance{LiveJournal}, \instance{Orkut}, \instance{Friendster}, \instance{Twitter}), \emph{collaboration}~(\instance{AstroPh}), and web~(\instance{Slovakia}, \instance{Slovakia$^\top$}) networks. \instance{Slovakia$^\top$} is obtained from~\instance{Slovakia} by reversing all arcs~(influence follows the reverse direction of links).

Kempe et al.~\cite{KKT:KDD2003} proposed two natural ways of associating probabilities with edges in the binary IC model: the \emph{uniform} scheme assigns a constant probability~$p$ to each directed edge~(they used~$p=0.1$ and~$p=0.01$), whereas in the \emph{weighted cascade}~(wc) scheme the probability is the inverse of the degree of the head node~(making the probability that a node is influenced less dependent on its number of neighbors). We consider the wc scheme by default, but we will also experiment with the uniform scheme~(un). These two schemes are the most commonly tested in previous studies of scalability~\cite{Leskovec:KDD2007,CWY:KDD2009,CWW:KDD2010,JHC:ICDM2012,OAYK:AAAI2014,TXS:sigmod2014}.

\subsection{Influence Maximization}

\begin{table*}[tb] \centering \caption{\label{tab:main} Performance of \skim and \irie. \skim uses~$k=64$, $\ell=64$, and we evaluate the influence on~512~(different) sampled instances. For all runs~(except those for~$n$ seeds) we set a time limit of two hours. For the runs that did not finish~(DNF), we report the influence of the seed set~(its size is shown in parenthesis after ``DNF'') computed within the time limit~(*).} \setlength{\tabcolsep}{1.5ex} \begin{tabular}{@{}lrrrrrrrrrrr@{}} \toprule &&& \multicolumn{4}{c}{\textbf{influence [\%]}} & \multicolumn{5}{c}{\textbf{running time [sec]}}\\ \cmidrule(lr){4-7}\cmidrule(l){8-12} &&& \multicolumn{2}{c}{50 seeds} & \multicolumn{2}{c}{1000 seeds} & \multicolumn{2}{c}{50 seeds} & \multicolumn{2}{c}{1000 seeds} & $n$ seeds \\ \cmidrule(lr){4-5}\cmidrule(lr){6-7}\cmidrule(lr){8-9}\cmidrule(lr){10-11}\cmidrule(l){12-12} instance & $|V|$ [$\cdot 10^3$] & $|A|$ [$\cdot 10^3$] & \skim & \irie & \skim & \irie & \skim & \irie & \skim & \irie & \skim\\ \midrule \instance{AstroPh} & 14.8 & 239.3 & 11.1 & 11.4\phantom{*} & 45.9 & 46.5\phantom{*} & 0.5 & 0.5 & 1.0 & 4.3 & 1.9\\ \instance{Epinions} & 75.9 & 508.8 & 15.8 & 15.9\phantom{*} & 34.4 & 34.1\phantom{*} & 0.7 & 0.9 & 1.6 & 10.3 & 6.7\\ \instance{Slashdot} & 77.4 & 828.2 & 21.4 & 21.6\phantom{*} & 52.1 & 52.3\phantom{*} & 0.8 & 1.5 & 1.9 & 19.8 & 7.5\\ \instance{Gowalla} & 196.6 & 1\,900.7 & 18.1 & 18.1\phantom{*} & 30.9 & 31.1\phantom{*} & 1.4 & 5.1 & 3.5 & 75.2 & 21.5\\ \instance{TwitterFollowers} & 456.6 & 14\,855.9 & 4.4 & 4.2\phantom{*} & 17.2 & 17.5\phantom{*} & 3.0 & 23.1 & 10.7 & 388.5 & 85.1\\ \instance{LiveJournal} & 4\,847.6 & 68\,475.4 & 1.6 & 1.5\phantom{*} & 6.8 & 6.7\phantom{*} & 8.6 & 261.1 & 31.1 & 4\,576.5 & 933.0\\ \instance{Orkut} & 3\,072.6 & 234\,370.2 & 5.3 & 5.3\phantom{*} & 12.1 & 11.5\textsuperscript{*} & 34.0 & 473.3 & 102.9 & DNF~(915) & 1\,197.2\\ \instance{Friendster} & 65\,608.4 & 1\,806\,067.1 & 9.5 & 8.8\textsuperscript{*} & 15.4 & 8.8\textsuperscript{*} & 794.0 & DNF~(43) & 1\,308.5 & DNF~(43) & 19\,254.2\\ \instance{Twitter} & 41\,652.2 & 1\,468\,364.9 & 21.1 & 21.1\phantom{*} & 38.0 & 25.3\textsuperscript{*} & 965.4 & 4\,233.4 & 1\,912.8 & DNF~(92) & 11\,558.8\\ \instance{Slovakia} & 50\,636.2 & 1\,930\,292.9 & 5.4 & 4.8\phantom{*} & 14.8 & 10.1\textsuperscript{*} & 86.6 & 2\,272.5 & 293.9 & DNF~(290) & 11\,743.4\\ \instance{$\text{\textsf{Slovakia}}^\top$} & 50\,636.2 & 1\,930\,292.9 & 10.3 & 10.0\phantom{*} & 25.9 & 16.7\textsuperscript{*} & 220.7 & 1\,740.6 & 621.4 & DNF~(230) & 11\,679.3\\ \bottomrule \end{tabular} \end{table*}

This section evaluates \skim, our new sketch-based influence maximization algorithm. By default we set the number of sampled instances to~$\ell = 64$ and compute sketches with~$k = 64$ entries. (These choices will be justified in later experiments.) To evaluate the actual influence values of the seeds computed by \skim, we use a set of 512 \emph{different} sampled instances, in which we simply run BFSes a posteriori.

Table~\ref{tab:main} summarizes the performance of our algorithm on several networks of varying sizes with up to almost two billion edges. Besides the network sizes, the table reports results for three seed set sizes~$s$: 50,~1000 and~$n$, i.e., computing full permutation. In each case, it reports the total running time of our algorithm as well as the total influence of the related seed set as a percentage of~$n$. (Note that for~$s=n$ this value is~100\,\% by definition, so we omit it in the table.) For~$s=50$ and~$1000$, the table also reports the corresponding numbers for \irie~\cite{JHC:ICDM2012}, one of the fastest available heuristics that can generate full permutations. We use our own implementation of \irie, which is somewhat faster than the one evaluated in the original paper. Except for~$s=n$, we set an execution time limit of two hours; we report~``DNF'' and the corresponding number of computed seeds for those runs that did not finish.

The table shows that the influences computed by \irie and \skim are very close; sometimes \skim being better. However, \skim is significantly faster, outperforming \irie by several orders of magnitude on many instances. In particular, when computing~1000 instead of 50~seeds, \skim{}'s speedup over \irie becomes more evident as \irie's running time grows linearly with the number of seed nodes, whereas with \skim\ it decreases with the size of the residual problem. As a result, we can compute the~1000 most influential nodes on a graph with 65 million nodes and 1.8 billion edges~(Friendster) in just 22 minutes. Similarly, computing a full influence ordering with \skim takes less then 5.5 hours on all graphs.

\begin{table}[tb] \centering \caption{Comparing \skim and \timplus regarding influence and running time for 50 and 1000 seeds.\label{tab:tim}} \setlength{\tabcolsep}{.55ex} \begin{tabular}{@{}lrrrrrrrr@{}} \toprule & \multicolumn{4}{c}{\textbf{influence~[\%]}} & \multicolumn{4}{c}{\textbf{running time~[sec]}}\\ \cmidrule(lr){2-5}\cmidrule(l){6-9} & \multicolumn{2}{c}{50 seeds} & \multicolumn{2}{c}{1000 seeds} & \multicolumn{2}{c}{50 seeds} & \multicolumn{2}{c}{1000 seeds}\\ \cmidrule(lr){2-3}\cmidrule(lr){4-5}\cmidrule(lr){6-7}\cmidrule(l){8-9} instance & \skim & \tim & \skim & \tim & \skim & \tim & \skim & \tim\\ \midrule \instance{AstroPh} & 11.1 & 11.6 & 45.9 & 47.0 & 0.5 & 0.7 & 1.0 & 1.8\\ \instance{Epinions} & 15.8 & 15.8 & 34.4 & 34.8 & 0.7 & 0.3 & 1.6 & 2.3\\ \instance{Slashdot} & 21.4 & 22.2 & 52.1 & 52.6 & 0.8 & 1.2 & 1.9 & 6.9\\ \instance{Gowalla} & 18.1 & 18.2 & 30.9 & 31.4 & 1.4 & 2.0 & 3.5 & 13.4\\ \instance{TwitterF's} & 4.4 & 4.6 & 17.2 & 17.6 & 3.0 & 7.1 & 10.7 & 28.7\\ \instance{LiveJournal} & 1.6 & 1.7 & 6.8 & 7.0 & 8.6 & 26.0 & 31.1 & 89.1\\ \instance{Orkut} & 5.3 & 5.4 & 12.1 & 12.3 & 34.0 & 102.0 & 102.9 & 427.8\\ \instance{Friendster} & 9.5 & 9.6 & 15.4 & 15.6 & 794.0 & 406.5 & 1,308.5 & 410.5\\ \instance{Twitter} & 21.1 & 21.3 & 38.0 & 38.1 & 965.4 & 291.6 & 1,912.8 & 795.0\\ \instance{Slovakia} & 5.4 & 5.5 & 14.8 & 14.9 & 86.6 & 299.3 & 293.9 & 647.1\\ \instance{$\text{\textsf{Slovakia}}^\top$} & 10.3 & 10.5 & 25.9 & 26.3 & 220.7 & 384.1 & 621.4 & 313.4\\ \bottomrule \end{tabular} \end{table}

We also compare \skim to \timplus~\cite{TXS:sigmod2014}, the fastest influence maximization algorithm we are aware of. We ran their implementation~(kindly given to us by the authors) to report figures on our instances. As in their experiments, we set the~$\varepsilon$ parameter of \timplus to~$1.0$. Table~\ref{tab:tim} reports the influence~(as percentage of~$n$) as well as the running time for 50~and 1000~seed nodes. We note that \skim and \timplus are extremely close in quality, with \timplus tending to be slightly better. \skim is faster than \timplus on most instances except on \instance{Friendster}, \instance{Twitter}, and \instance{$\text{\textsf{Slovakia}}^\top$} with 1000~seeds, and generally the two are never more than a factor of three apart. However, recall that \skim actually computes a \emph{sequence} of nodes such that every prefix of this sequence also~(approximately) maximizes the influence. In contrast, \timplus\ must be rerun to obtain a smaller set of maximally influential nodes.

We next argue why our paremeter choices are reasonable. First, we evaluate the impact of the number $\ell$ of instances on the solution quality. Figure~\ref{fig:simulations} (left) reports the quality of the seed nodes found by \textsc{Greedy} (\greedy) when we use different~$\ell$ values during the algorithm, but evaluate the quality of the resulting seed set on~4096~(different) instances. We observe that increasing~$\ell$ does help quality, but only up to a certain point. In particular, values beyond 64 yield modest improvements. Since our running times depend on~$\ell$, we use this value by default.

\begin{figure}[tb] \centering \includegraphics[page=1]{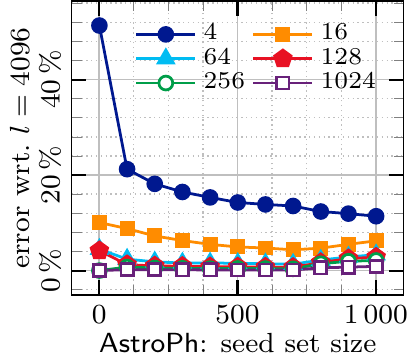}\hfil \includegraphics{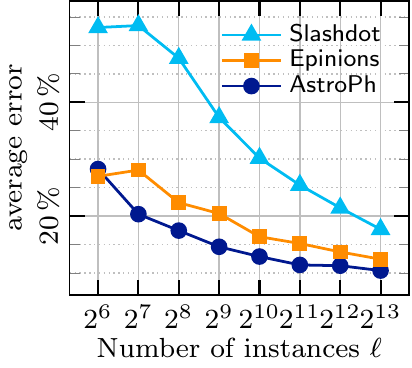} \caption{Evaluating different numbers of simulations~(left) and evaluating the average error of our oracle on 1000 random seeds, subject to varying~$\ell$. The right plot is discussed in Section~\ref{sec:exp:oracles}.} \label{fig:simulations} \label{fig:ell} \end{figure}

Figure~\ref{fig:parameters} compares \skim to \greedy, \irie, and \degree~(including nodes by order of decreasing degree) on two inputs: \instance{Slashdot} and \instance{TwitterFollowers}. For \skim, we test various values for~$k$~(4, 16, 64, 256). We report the influence error when compared to \greedy~(top) and the running time~(bottom). We observe that the error for \skim decreases as we increase~$k$, $k=64$ being the sweet spot, after which solution quality does not improve by much anymore. Running times increase for all algorithms with the size of the seed set, but \skim is consistently the fastest algorithm for any size.

\begin{figure}[tb] \centering \includegraphics[page=2]{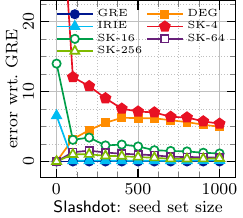}\hfil\includegraphics[page=3]{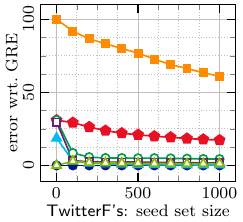}

\medskip

\includegraphics[page=2]{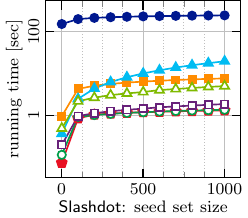}\hfil\includegraphics[page=3]{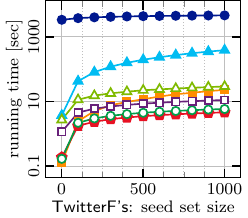} \caption{Evaluating influence and running time for several algorithms. The legend applies to all plots.} \label{fig:parameters} \end{figure}

Figure~\ref{fig:models} evaluates the performance of \skim\ and \irie on the two IC schemes~(wc,un), using \instance{TwitterFollowers} as input. We observe that \skim matches the solution quality of \irie but is significantly faster.

\begin{figure}[tb] \centering \includegraphics[page=3]{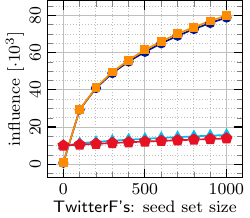}\hfil\includegraphics[page=3]{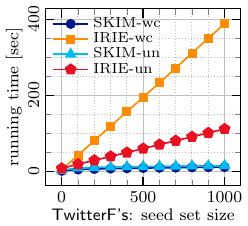} \caption{Evaluating \skim and \irie on the uniform~(un) and weighted cascade~(wc) models. The legend applies to both plots.} \label{fig:models} \end{figure}

Finally, Figure~\ref{fig:coverage} shows the influence~(top) and running time~(bottom) of \skim when computing the full permutation. We plot the relative influence and running time~(both as percentage) subject to the number of computed seed nodes as the algorithm progresses~(also as percentage of~$n$). To the best of our knowledge, we are the first who are able to compute (approximately) the full Pareto front of influence versus seed set size on graphs with billions of edges within a few hours only. The tradeoff seems to characterize the core of the network: On \instance{Slovakia$^\top$} and \instance{Twitter},~0.1\% of the nodes already cover almost~50\% of the entire graph, while on \instance{Slashdot} and \instance{Friendster},~0.1\% of the seeds only cover~25--30\% of the graph, albeit with a faster growth. Other instances have a slower growth in influence, but on all instances~10\% of the nodes cover at least~50\% of the graph. Regarding running time, we observe that all instances exhibit similar behavior. In particular, more than 50\% of the total running time is spent computing the first 10\% of seed nodes.

\begin{figure}[tb] \centering\includegraphics[page=1]{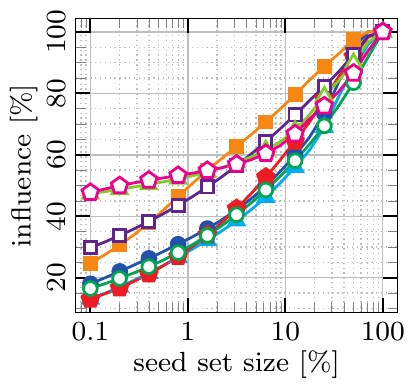}\hfil\includegraphics[page=2]{10.pdf}\caption{Evaluating influence permutations~(top) and running time~(bottom) on several instances. The legend applies to both plots.} \label{fig:coverage} \end{figure}

\subsection{Influence Oracles}\label{sec:exp:oracles}

This section evaluates our influence oracle~(cf.~Section~\ref{binaryQ:sec}). We use the IC model (with wc probabilities) to generate a set of~$\ell=64$ instances. We build combined reachability sketches of size~$k=64$ for this set of instances and evaluate the performance of our oracle~(cf.~Section~\ref{model:sec}).

Table~\ref{tab:oracle} summarizes the performance of our oracle on several networks. It reports the time spent for preprocessing and the required space~(in~MiB) to store the combined sketches. Queries are evaluated for seed set sizes~$s$ of 1, 50, and 1000. For each~$s$, we generate 100 seed sets whose nodes are selected uniformly at random. We report the average running time of the query~(estimator) in microseconds and the relative error of the estimated influence when compared to the exact influence of the respective seed set.

We observe that preprocessing times are reasonable for all graphs while space consumption is essentially linear in the number of nodes. For example, on \instance{LiveJournal}~(the biggest instance tested), the sketches require~2.3\,GiB of space, which we computed in just 34 minutes. The influence of a single node can then be estimated in 1--2\,\textmu{}s, while for 1000 seed nodes we require 5.2\,ms. Note that the query time is almost independent of the graph size. Using~$k=64$, the error stays well below 10\% for one seed node, and decreases significantly for larger seed sets~(to around~1\% for~$s=1000$).

\begin{table}[tb] \centering \caption{Evaluating our influence oracle with~$\ell=64$.\label{tab:oracle}} \setlength{\tabcolsep}{0.6ex} \begin{tabular}{@{}lrrrrrrrr@{}} \toprule & \multicolumn{2}{c}{\textbf{preproc.}} & \multicolumn{6}{c}{\textbf{queries}} \\ \cmidrule(lr){2-3}\cmidrule(l){4-9} & & & \multicolumn{2}{c}{1 seed} & \multicolumn{2}{c}{50 seeds} & \multicolumn{2}{c}{1000 seeds}\\ \cmidrule(lr){4-5}\cmidrule(lr){6-7}\cmidrule(l){8-9} & time & space & time & err. & time & err. & time & err. \\ instance & [sec] & [MiB] & [\textmu{}s] & [\%] & [\textmu{}s] & [\%] & [\textmu{}s] & [\%] \\ \midrule \instance{AstroPh} & 4 & 7.2 & 1.6 & 8.5 & 166.7 & 2.1 & 4\,658.3 & 0.5\\ \instance{Epinions} & 10 & 37.1 & 1.3 & 5.2 & 155.0 & 3.4 & 5\,011.1 & 1.1\\ \instance{Slashdot} & 20 & 37.8 & 1.5 & 6.0 & 155.2 & 3.9 & 4\,982.3 & 1.0\\ \instance{Gowalla} & 46 & 96.0 & 1.5 & 7.3 & 179.8 & 3.2 & 5\,275.6 & 1.1\\ \instance{TwitterFollowers} & 229 & 223.0 & 2.1 & 7.0 & 190.2 & 3.3 & 5\,061.8 & 0.8\\ \instance{LiveJournal} & 2\,064 & 2\,367.0 & 2.0 & 7.1 & 189.6 & 3.0 & 5\,168.3 & 0.9\\ \bottomrule \end{tabular} \end{table}

Figure~\ref{fig:oracleerror} shows in detail how the error of the estimator~($y$ axis) decreases when the seed set size increases~($x$ axis). To better evaluate the performance of estimating the \emph{union} of several reachability sets, we use the following \emph{neighborhood generator} for queries: For each query, it first picks a node~$u$ at random with probability proportional to its degree. From~$u$ it exhaustively grows a BFS of the smallest depth~$l$ such that the tree contains at least~$s$ nodes. The nodes for the seed set are then uniformly sampled from this tree. With this generator, we expect the reachability sets of seed nodes to highly overlap. Looking at the figure, we observe that the estimation error of our oracle decreases rapidly for increasing~$s$. Also, running queries from the neighborhood generator~(right) compared to the uniform one~(left), has almost no effect on the estimation error; for 50 seed nodes it is even better on many instances.

\begin{figure}[tb] \centering\includegraphics[page=1]{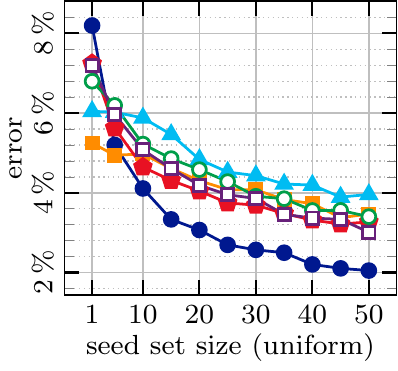}\hfil\includegraphics[page=2]{12.pdf} \caption{Evaluating our oracle for seed sets of varying size, which are selected uniformly at random~(left) or with our BFS-based method~(right).\label{fig:oracleerror}} \end{figure}

Finally, Figure~\ref{fig:ell}~(right) reports the performance of the oracle for fixed instances on the general IC model. We vary the number~$\ell$ of instances generated by simulations when building the oracle, but compute the error on a different set of~8192 instances. Since our oracle implementation is optimized for fixed instances, we see a higher error with~$\ell=64$. We can also see that the error decreases with the number of simulations. We conclude that for an IC model oracle, it is beneficial to construct sketches that have approximation guarantees with respect to the IC model itself~(cf. Section~\ref{ICsketches:Sec}) rather than work with simulations.

\section{Conclusion}

We presented highly scalable algorithms for binary influence computation. \skim\ is a sketch-space implementation of the greedy influence maximization algorithm that scales it by several orders of magnitude, to graphs with billions of edges. \skim\ computes a sequence of nodes such that each prefix has a probabilistic guarantee on approximation quality that is close to that of {\sc Greedy}. We also presented sketch-based influence oracles, which after a near-linear processing of the instances can estimate influence queries in time proportional to the number of seeds. Our experimental study focused on instances generated by an IC model, since the fastest algorithms we compared with only apply in this model. Our experiments revealed that \skim\ is accurate and faster than other algorithms by one to two order of magnitude.

In future work, we plan to develop a {\skim}-like algorithm for {\em timed influence}, where edges have lengths that are interpreted as transition times and we consider both the speed and scope of infection \cite{Gomez-RodriguezBS:ICML2011,CLZ:AAAI2012,CDFGGW:COSN2013,ACKP:KDD2013,DSGZ:nips2013}. We also plan to use sketches to efficiently estimate the Jaccard similarity of the influence sets of two nodes, which we believe to be an effective similarity measure~\cite{CDFGGW:COSN2013}.

\bibliographystyle{plain} 

\end{document}